\documentclass[letterpaper,10pt,conference]{ieeeconf}
\ifdefined\pdfminorversion
\fi
\IEEEoverridecommandlockouts
\usepackage{amsmath,amssymb,mathtools,bm}

\usepackage{amsthm}
\usepackage{graphicx}
\usepackage{booktabs}
\usepackage{multirow}
\usepackage{array}
\usepackage{url}
\usepackage{cite}
\usepackage{xcolor}
\usepackage{siunitx}
\let\labelindent\relax
\usepackage{enumitem}

\usepackage[protrusion=false]{microtype}
\usepackage{placeins}
\usepackage{flushend}
\usepackage{algorithm}
\usepackage{algorithmic}

\ifdefined
\fi

\graphicspath{{figures/}}
\DeclareMathOperator*{\argmax}{arg\,max}
\DeclareMathOperator{\diag}{diag}
\newtheorem{theorem}{Theorem}
\newtheorem{proposition}{Proposition}
\newtheorem{remark}{Remark}
\newtheorem{assumption}{Assumption}

\title{\LARGE \bf
Gamma–Laplace Surrogate for Variance-Aware Sensor Placement for Detecting Poisson Distributed Targets}

\author{Md Muhtasim Fuad, Mingyu Kim$^*$, Daniel J.~Stilwell, Jorge Jimenez%
\thanks{Md Muhtasim Fuad and Mingyu Kim are with the Department of Electrical and Computer Engineering, Georgia Southern University, Statesboro, GA, USA
        {\tt\small \{mf21842,mkim\}@georgiasouthern.edu}}%
\thanks{Daniel J.~Stilwell is with the Department of Electrical and Computer Engineering, Virginia Tech, Blacksburg, VA, USA
        {\tt\small stilwell@vt.edu}}%
\thanks{Jorge Jimenez is with Johns Hopkins University Applied Physics Laboratory, Laurel, MD, USA
        {\tt\small jorge.jimenez@jhuapl.edu}}%
\thanks{$^*$Corresponding author}%
}

\begin{document}
\maketitle
\thispagestyle{empty}\pagestyle{empty}

\begin{abstract}
Sensor placement for stochastically arriving targets is studied using a void-probability objective, defined as the probability that no target remains undetected over a finite horizon. Direct optimization is intractable because it requires an expectation over a random intensity field. A common surrogate based on Jensen’s inequality replaces the random field with its mean, yielding a tractable objective but discarding distributional information. The proposed \emph{Gamma–Laplace surrogate} approximates cellwise exposure variables with moment-matched Gamma distributions and evaluates the objective using their closed-form Laplace transforms. Unlike scalar moment corrections or local expansions, this approach operates at the distribution level and retains a Laplace-transform structure under an independent-Gamma approximation. A log transformation and constant shift yield an objective shown to be monotone submodular, enabling efficient greedy optimization with performance guarantees. Numerical experiments on real ship-traffic data show consistency with the Jensen surrogate in low-uncertainty regimes. Under moderate uncertainty, the method improves candidate-location rankings, leading to more robust sensor selections and reduced sensitivity to variability in target arrivals.
\end{abstract}

\section{Introduction}

Sensor placement supports stochastic target detection in surveillance, environmental monitoring, and autonomous systems. In many applications, uncertain target arrivals are modeled using spatial point processes, including Poisson processes and log-Gaussian Cox processes (LGCPs) \cite{moller1998log, diggle2013statistical}. The \emph{void probability}, defined as the probability that no targets remain undetected, directly measures sensing effectiveness under uncertainty and provides an objective for sensor placement \cite{kim2023optimal,kim2025variance,kim2025outlier}.

Direct optimization of the void probability is computationally intractable because it requires the expectation of an exponential functional of a random variable. Prior work has therefore used tractable surrogate objectives. Jensen’s inequality yields a mean-based lower bound that enables monotone submodular optimization and efficient greedy algorithms with a $(1-1/e)$ guarantee \cite{kim2023optimal, nemhauser1978analysis, krause2008near}. However, this approximation depends only on the mean and ignores uncertainty, which can lead to suboptimal sensor placement in heterogeneous or high-variance environments.

Variance-aware approximations based on second-order Taylor expansions have been proposed to incorporate uncertainty \cite{kim2025variance}. They approximate the objective function locally using aggregate mean and variance. The placement-dependent variance correction can oppose the mean term, so the monotonicity and submodularity argument for the mean-only objective does not carry over. Section~\ref{sec:scalar} explains this limitation; Section~\ref{sec:numerical} evaluates candidate-ranking effects.

This paper addresses this limitation with a \emph{Gamma--Laplace surrogate}, a distributional approximation of the exposure random variable. It approximates the distribution of the aggregated exposure using moment-matched Gamma variables and evaluates the resulting \emph{Laplace transform} in a closed form. The approach incorporates both mean and variance while preserving the exponential transform structure of the original objective. The resulting surrogate therefore admits a monotone and submodular transformed objective, enabling efficient greedy optimization with theoretical $(1-1/e)$ guarantees. It also recovers the lower bound in Jensen's inequality in the zero-variance limit, providing a consistent generalization of mean-based methods.

The framework is evaluated on ship-traffic data from \cite{boem_noaa_accessais}. The posterior intensity is estimated using integrated nested Laplace approximation (INLA) \cite{rue2009approximate,martins2013bayesian,juan2021enhancing}, a deterministic Bayesian inference method for latent Gaussian models that provides accurate and computationally efficient approximations to posterior marginals.

\textbf{Related Work}
Submodular optimization is commonly used to formulate sensor placement and experimental design problems because it offers computational efficiency and provable approximation guarantees \cite{nemhauser1978analysis, krause2008near}. These approaches have been widely applied in Gaussian process sensing, environmental monitoring, and active learning \cite{krause2007nonmyopic, chepuri2015sparse}. Poisson and Cox process models are standard tools for representing spatially distributed targets with uncertain intensities in stochastic detection settings \cite{diggle2013statistical, moller1998log,jimenez2025optimal}.

Laplace-transform approximations for random variables have been studied extensively in applied probability. Gamma distributions often provide tractable moment-matching approximations because their Laplace transforms have closed forms \cite{asmussen2007stochastic, haenggi2012stochastic}. These approximations incorporate distributional information while retaining computational tractability.

Surrogates for void-probability optimization trade off tractability and fidelity. The Jensen-based method in \cite{kim2023optimal} is tractable and preserves submodularity but ignores uncertainty. The Taylor-based method in \cite{kim2025variance} incorporates variance locally but does not inherit the mean-only greedy guarantee. The proposed Gamma--Laplace surrogate incorporates distributional information while retaining a structure compatible with greedy optimization.

\subsection*{Contributions}
This paper addresses the tradeoff between approximation fidelity and optimization tractability in void-probability sensor placement: 
i) a Gamma--Laplace surrogate is derived by modeling cellwise integrated intensities with moment-matched Gamma distributions, yielding a closed-form Laplace-transform expression that incorporates both mean and variance, recovers the lower bound in Jensen's inequality in the zero-variance limit, and avoids the locality of Taylor-based approximations; 
ii) after a log transformation and constant shift, the surrogate objective is shown to be monotone and submodular, enabling efficient greedy optimization with a $(1-1/e)$ approximation guarantee; 
iii) experiments on real ship-arrival data demonstrate that the surrogate matches baseline performance in low-uncertainty regimes and improves candidate ranking and greedy selection decisions under moderate uncertainty.

Section~II formulates the problem and reviews Jensen-based surrogates. Section~III introduces the Gamma--Laplace surrogate and its properties. Section~IV analyzes the surrogate and its optimization. Section~V describes the dataset, experimental setup, and numerical results. Section~VI discusses implications and limitations. Section~VII concludes the paper.

\section{Problem Formulation}
\label{sec:problem}
The stochastic sensor placement problem combines a probabilistic target arrival model, a sensing model, and a void-probability objective.

\subsection{Target Arrival Model}

A spatial domain \(\Psi \subset \mathbb{R}\) is considered in which targets arrive according to a stochastic intensity field \(\lambda(s)\), where \(s \in \Psi\). To capture spatial uncertainty and heterogeneity, the intensity is modeled as a log-Gaussian Cox process (LGCP)
\begin{align*}
\log \lambda(s) \sim \mathcal{GP}(\mu_\lambda(s), \Sigma_\lambda(s,s'))
\end{align*}
where \(\mu_\lambda(s)\) and \(\Sigma_\lambda(s,s')\) denote the mean and covariance functions of the latent Gaussian field and $s,s'\in \Psi$.

\subsection{Sensor Detection and Placement Notation}

A sensor network is deployed over \(\Psi\) to detect targets under the stochastic arrival model. The finite candidate set of admissible sensor locations is denoted by \(\mathcal{V} \subseteq \Psi\), and the selected subset by \(\mathcal{A} \subseteq \mathcal{V}\), with \(|\mathcal{A}| \le K\). A sensor at \(a_i \in \mathcal{V}\) detects a target at location \(s\) with probability \(\gamma(s,a_i)\).

Assuming conditional independence of detection failures across sensors, a target at \(s\) is missed by all sensors in \(\mathcal{A}\) with probability $\pi(s,\mathcal{A}) = \prod_{a_i \in \mathcal{A}} \left(1-\gamma(s,a_i)\right)$.

\subsection{Void Probability Objective}

For a planning horizon $T$, define the aggregate missed exposure
\[
Y(\mathcal A)=T\int_\Psi\lambda(s)\,\pi(s,\mathcal A)\,ds.
\]
Here $\lambda(s)$ is an arrival rate per unit distance and time, so $Y(\mathcal A)$ is the conditional expected number of missed targets given the intensity field, not a realized count. Conditional Poisson thinning gives a probability $e^{-Y(\mathcal A)}$ of no missed targets. Averaging over the random intensity field yields the void probability
\[
\nu(\mathcal A)=\mathbb E[e^{-Y(\mathcal A)}].
\]

\subsection{Discretization and Aggregate Moments}

Discretize $\Psi$ into $G$ cells with representative locations $s_g$ and spatial weights $\Delta_g$. Define the cell exposure $X_g:=T\Delta_g\lambda(s_g)$ and the miss probability $\pi_g(\mathcal A):=\pi(s_g,\mathcal A)$. The discretized aggregate is
\begin{align}
Y_g(\mathcal A)&:=\sum_{g=1}^G X_g\,\pi_g(\mathcal A),\label{eq:aggregate_Y}\\
\nu_g(\mathcal A)&:=\mathbb E[e^{-Y_g(\mathcal A)}]\approx\nu(\mathcal A).
\label{eq:discrete_void}
\end{align}
The subscript in $Y_g$ denotes the \emph{grid approximation of the entire aggregate}, not a single cell or the Gamma model; the summation index $g$ in $X_g$ denotes a cell. Both $X_g$ and $Y_g$ are random conditional expected counts induced by the intensity field.

Write $\mu_g=\mathbb E[X_g]$, $v_g=\operatorname{Var}(X_g)$, and $C_{gh}=\operatorname{Cov}(X_g,X_h)$. The grid aggregate moments are
\begin{align*}
\mu_{Y_g}(\mathcal A)&:=\mathbb E[Y_g(\mathcal A)]=\sum_g\mu_g\pi_g(\mathcal A),\\
\sigma_{Y_g}^2(\mathcal A)&:=\operatorname{Var}(Y_g(\mathcal A))\\
&=\sum_g v_g\pi_g(\mathcal A)^2
 +\sum_{g\ne h}C_{gh}\pi_g(\mathcal A)\pi_h(\mathcal A).
\end{align*}
Denoting the existing finite sum by $Y_g(\mathcal A)$ does not change calculations on the same grid. Discretization error is separate from Gamma moment matching and the independence approximation; its numerical magnitude is not quantified here.

\subsection{Scalar Surrogates Used for Comparison}
\label{sec:scalar}

For the continuous exposure, write $\mu_Y=\mathbb E[Y]$ and $\sigma_Y^2=\operatorname{Var}(Y)$. Jensen's inequality gives $e^{-\mu_Y}\le\nu$, while the second-order Taylor approximation is $e^{-\mu_Y}(1+\sigma_Y^2/2)$ \cite{kim2025variance}. On the grid used by \emph{all} three methods, these become
\begin{align*}
\nu_{\mathrm J}(\mathcal A)&=e^{-\mu_{Y_g}(\mathcal A)}\le\nu_g(\mathcal A),\\
\nu_{\mathrm{MV}}(\mathcal A)&=e^{-\mu_{Y_g}(\mathcal A)}
 \left(1+\tfrac12\sigma_{Y_g}^2(\mathcal A)\right).
\end{align*}
Neither formula requires independent cells. The Taylor correction is local and can exceed one. Its variance term can also counteract the mean term: for one cell with miss probability $t$, mean $\mu$, and variance $v$, its derivative is
\[
\frac{d}{dt}\left[e^{-\mu t}(1+vt^2/2)\right]
=e^{-\mu t}\left[vt-\mu(1+vt^2/2)\right].
\]
The bracketed term need not be nonpositive. Thus decreasing the miss probability need not increase the Taylor surrogate, and the classical monotone-submodular greedy guarantee does not follow. No such guarantee is asserted here for $\nu_{\mathrm{MV}}$.

\section{Gamma--Laplace Surrogate}
\label{sec:gamma_laplace}

Unlike $\nu_{\mathrm J}$ and $\nu_{\mathrm{MV}}$, the proposed construction approximates cellwise distributions and their Laplace factors rather than applying a local correction to the scalar objective function. The physical aggregate $Y$ and its grid approximation $Y_g$ denote the same quantities for all methods.

\begin{assumption}[Cellwise factorization]
\label{ass:factorization}
The cell exposures $X_1,\ldots,X_G$ are treated as independent only when constructing the Gamma--Laplace surrogate.
\end{assumption}
This is a surrogate-level approximation, not a claim about the correlated posterior field. Finer discretization does not make independence more plausible: with a fixed Mat\'ern correlation length $\ell$, adjacent cell centers approach each other as $\Delta\to0$, and their latent-field correlation approaches one, not zero. A coarser grid may reduce short-range dependence but also loses spatial resolution; discretization alone does not establish independence.

At the second-moment level, factorization retains the cell means, marginal variances, and aggregate mean but omits the cross-cell covariance contribution to $\sigma_{Y_g}^2$. It also omits higher-order dependence. The resulting diagonal variance is
\[
\sigma_{\mathrm{diag}}^2(\mathcal A):=\sum_gv_g\pi_g(\mathcal A)^2,
\]
which is not generally $\sigma_{Y_g}^2(\mathcal A)$. Consequently, the surrogate's dependence error need not disappear with grid refinement. Monte Carlo evaluation in Section~\ref{sec:numerical} retains correlated posterior samples, and the reported grid-sensitivity study assesses performance, not the validity of independence.

For the surrogate only, each cell exposure is matched to a Gamma distribution with shape $k_g$ and scale $\theta_g$:
\begin{align}
k_g=\frac{\mu_g^2}{v_g},\qquad
\theta_g=\frac{v_g}{\mu_g}.
\label{eq:gamma_params}
\end{align}
For $\mu_g,v_g>0$, the Gamma Laplace transform is
\begin{align}
\mathbb E_{\Gamma}[e^{-tX_g}]=(1+\theta_gt)^{-k_g},\quad t\ge0.
\label{eq:gamma_laplace}
\end{align}
Here $\mathbb E_\Gamma$ denotes expectation under the matched Gamma law, not the original posterior. Applying \eqref{eq:gamma_laplace} at $t=\pi_g(\mathcal A)$ and factorizing gives
\begin{align}
\nu_{\mathrm{GL}}(\mathcal A)=
\prod_{g=1}^G(1+\theta_g\pi_g(\mathcal A))^{-k_g}.
\label{eq:gamma_product}
\end{align}
When $v_g=0$, use the limit $e^{-\mu_g\pi_g}$; a zero-mean nonnegative exposure contributes a factor of one.

\subsection*{Remarks on the Gamma--Laplace Surrogate}
The product in \eqref{eq:gamma_product} lies in $(0,1]$ and is a closed-form Laplace transform under the independent-Gamma model. It is not an exact evaluation of the correlated posterior objective. Moment matching retains the marginal means and variances; higher marginal moments are prescribed by the Gamma family, not recovered from the data. Unlike a single scalar moment correction, the cellwise construction preserves the spatial allocation of marginal uncertainty and supports the optimization structure proved below.

\section{Analysis and Optimization of the Proposed Surrogate}
\label{sec:theory}

The proposed Gamma--Laplace surrogate is analyzed in terms of (i) consistency with the Jensen approximation, (ii) its probabilistic interpretation under an independent-Gamma model, and (iii) the induced optimization structure that enables efficient sensor placement. All results are stated for the surrogate model and do not directly apply to the exact correlated posterior objective.

\subsection{Consistency in low-uncertainty regimes}
The proposed surrogate recovers the lower bound in Jensen's inequality when marginal variances vanish on a fixed grid.

\begin{proposition}[Zero-variance limit]
\label{prop:zero_var}
Fix \(\mu_g>0\) and \(t\in[0,1]\). If \(v_g\to 0\), then
\begin{align}
\lim_{v_g\to 0}(1+\theta_g t)^{-k_g}=e^{-\mu_g t}
\end{align}
Consequently,
\begin{align}
\lim_{\{v_g\}\to 0}\nu_{\mathrm{GL}}(\mathcal{A})=\nu_{\mathrm{J}}(\mathcal{A})
\end{align}
\end{proposition}

\begin{proof}
From \eqref{eq:gamma_params} and $\log(1+x)=x+\mathcal O(x^2)$,
\[
\log(1+\theta_gt)^{-k_g}
=-\frac{\mu_g^2}{v_g}\log\!\left(1+\frac{v_g}{\mu_g}t\right)
\longrightarrow-\mu_gt.
\]
Exponentiating and multiplying over the fixed $G$ cells gives $\nu_{\mathrm{GL}}\to e^{-\sum_g\mu_g\pi_g}=\nu_{\mathrm J}$, the lower bound in Jensen's inequality on the same grid.
\end{proof}

\subsection{Distributional interpretation}
The surrogate admits an exact interpretation under an independent-Gamma model.

\begin{proposition}[Exactness under independent-Gamma model]
\label{prop:gamma_exact}
Under Assumption~\ref{ass:factorization}, if each $X_g$ follows the matched Gamma law, then $\mathbb E_\Gamma[e^{-Y_g(\mathcal A)}]=\nu_{\mathrm{GL}}(\mathcal A)$ for every placement $\mathcal A$.
\end{proposition}
\begin{proof}
Independence and \eqref{eq:gamma_laplace} give
\begin{align*}
\mathbb E_\Gamma[e^{-Y_g(\mathcal A)}]
&=\prod_{g=1}^G\mathbb E_\Gamma[e^{-X_g\pi_g(\mathcal A)}]\\
&=\prod_{g=1}^G(1+\theta_g\pi_g(\mathcal A))^{-k_g}
=\nu_{\mathrm{GL}}(\mathcal A).\qedhere
\end{align*}
\end{proof}
This exactness applies only to the independent-Gamma model, not to the correlated posterior exposure law.

\subsection{Optimization structure}
The surrogate defines a tractable objective for sensor placement.

\begin{theorem}[Transformed monotone-submodular objective]
\label{thm:submodularity}
For $0\le\gamma(s_g,v)<1$, define
\begin{align*}
\ell_{g,v}:=\log\!\left(1-\gamma(s_g,v)\right)\le 0,\qquad
T_g(\mathcal{A}) := \sum_{v\in\mathcal{A}}\ell_{g,v}
\end{align*}
so that \(\pi_g(\mathcal{A})=e^{T_g(\mathcal{A})}\). Let
\begin{align*}
F(\mathcal{A})
=
-\log \nu_{\mathrm{GL}}(\mathcal{A})
=
\sum_{g=1}^{G}k_g\log\!\left(1+\theta_g e^{T_g(\mathcal{A})}\right)
\end{align*}
and define \(H(\mathcal{A}) := F(\emptyset)-F(\mathcal{A}) \ge 0\).
Then \(H(\mathcal{A})\) is monotone nondecreasing and submodular.
\end{theorem}

\begin{proof}[Proof sketch]
Write \(F(\mathcal{A})=\sum_{g=1}^{G}\psi_g\!\big(T_g(\mathcal{A})\big)\), where \(\psi_g(t):=k_g\log(1+\theta_g e^{t})\). Each \(T_g\) is \emph{modular}---a sum over the elements of \(\mathcal{A}\) with weights \(\ell_{g,v}=\log(1-\gamma(s_g,v))\le 0\). Differentiation gives
\[
\psi_g'(t)=\frac{k_g\theta_g e^{t}}{1+\theta_g e^{t}}>0,\qquad
\psi_g''(t)=\frac{k_g\theta_g e^{t}}{(1+\theta_g e^{t})^{2}}>0,
\]
so \(\psi_g\) is strictly increasing and convex.

\emph{Submodularity.} Fix \(S\subseteq\mathcal{V}\) and distinct \(a,b\notin S\), and set \(x=T_g(S)\), \(\alpha=\ell_{g,a}\), \(\beta=\ell_{g,b}\). By the modularity of \(T_g\), the discrete second difference of \(\psi_g\circ T_g\) at \(S\) equals
\begin{align*}
&\psi_g(x+\alpha+\beta)-\psi_g(x+\alpha)-\psi_g(x+\beta)+\psi_g(x)\\
&\hspace{1em}=\int_\alpha^0\!\int_\beta^0\psi_g''(x+u+v)\,dv\,du\ge0.
\end{align*}
Both integration intervals have nonnegative length because $\alpha,\beta\le0$. Thus $\psi_g\circ T_g$ is supermodular. Summing over $g$ makes $F$ supermodular and $H=F(\emptyset)-F$ submodular.

\emph{Monotonicity.} For \(v\notin\mathcal{A}\), \(T_g(\mathcal{A}\cup\{v\})=T_g(\mathcal{A})+\ell_{g,v}\le T_g(\mathcal{A})\); since \(\psi_g\) is increasing, \(F(\mathcal{A}\cup\{v\})\le F(\mathcal{A})\), i.e., \(H(\mathcal{A}\cup\{v\})\ge H(\mathcal{A})\). Also $H(\emptyset)=0$. The case of perfect detection ($\gamma=1$) follows by continuity. \qedhere
\end{proof}

\begin{remark}[Scope of guarantees]
The submodularity and approximation guarantees apply to the surrogate objective \(H(\mathcal{A})\) under Assumption~\ref{ass:factorization} and do not extend to the exact objective induced by the correlated posterior field.
\end{remark}

\subsection{Greedy Optimization for Sensor Placement}
\label{sec:algorithm}

The surrogate objective in Theorem~\ref{thm:submodularity} is optimized under the cardinality constraint \(|\mathcal{A}| \le K\) using the greedy rule
\begin{align}
a_{k+1}
=
\argmax_{v \in \mathcal{V}\setminus \mathcal{A}_k}
H(\mathcal{A}_k \cup \{v\})
\end{align}
with initialization \(\mathcal{A}_0=\emptyset\) and updates \(\mathcal{A}_{k+1}=\mathcal{A}_k\cup\{a_{k+1}\}\) for \(k=0,1,\dots,K-1\). The guarantee is specifically
\[
H(\mathcal A_K)\ge(1-1/e)\max_{|\mathcal A|\le K}H(\mathcal A)
\]
\cite{nemhauser1978analysis}, not a multiplicative guarantee for $\nu$ or $\nu_{\mathrm{GL}}$. Efficient evaluation is enabled by the update
\begin{align*}
\pi_g(\mathcal{A}\cup\{a\})
=
\pi_g(\mathcal{A})\left(1-\gamma(s_g,a)\right)
\end{align*}
which allows incremental computation of marginal gains across candidate locations. Algorithm~\ref{alg:greedy} summarizes the full procedure.

\begin{algorithm}[H]
\caption{Greedy sensor placement for maximizing the surrogate objective \(H(\mathcal{A})\)}
\label{alg:greedy}
\begin{algorithmic}[1]
\STATE Initialize \(\mathcal{A}_0 \gets \emptyset\)
\FOR{\(k = 0,1,\dots,K-1\)}
    \STATE \(a_{k+1} \gets \arg\max\limits_{v \in \mathcal{V} \setminus \mathcal{A}_k}
    H(\mathcal{A}_k \cup \{v\})\)
    \STATE \(\mathcal{A}_{k+1} \gets \mathcal{A}_k \cup \{a_{k+1}\}\)
\ENDFOR
\STATE \textbf{return} \(\mathcal{A}_K\)
\end{algorithmic}
\end{algorithm}

\section{Numerical Results}
\label{sec:numerical}

This section compares the proposed Gamma--Laplace surrogate with two baselines: the Jensen mean-based surrogate and a variance-incorporating Jensen approximation based on a second-order Taylor expansion \cite{kim2025variance}. The comparison uses a ship-arrival dataset and examines (i) final sensor-placement performance evaluated using posterior Monte Carlo, (ii) greedy-path behavior during sensor selection, and (iii) surrogate fidelity in terms of calibration and ranking accuracy.

The one-dimensional example isolates the effect of uncertainty on greedy decisions, final placement performance, and surrogate fidelity. The framework also accommodates higher-dimensional discretizations.

\subsection{Sensor model}

The one-dimensional sensing model is
$\gamma(x,a_i)
=
\rho \exp\!\left(-\frac{(x-a_i)^2}{\sigma_\ell}\right)$
where $\rho = 0.95$ and $\sigma_\ell = 0.90$, with distances measured in kilometers. Candidate locations are defined on a uniform \SI{50}{m} grid over the spatial domain, and greedy selection is performed under a cardinality constraint of $K = 20$ sensors.

\subsection{Ship Arrival Data, Intensity Model, and Evaluation Setup}

\textbf{Data and geometry.}
The dataset consists of \emph{ship arrivals} rather than full trajectories; an extension to trajectory data is discussed in \cite{kim2025barrier}. All 785 observations lie on a vertical segment with endpoints \(P_1 = (-76.05,\,36.88)\) and \(P_2 = (-76.05,\,37.04)\) near Norfolk, Virginia, USA \cite{boem_noaa_accessais}. The segment length is approximately \SI{17.811}{km}, and the observed crossings span \SIrange{5.428}{17.777}{km} from \(P_1\). The recorded positions are consistent with this geometry to machine precision, so the problem reduces to a one-dimensional sensing domain along the coordinate \(x\), measured from \(P_1\). Fig.~\ref{fig:geometry} shows a heatmap of ship traffic near \(P_1\) and \(P_2\), supporting this line-segment representation. 

\begin{figure}[t]
\centering

\includegraphics[width=0.8\linewidth, height= 0.6\linewidth]{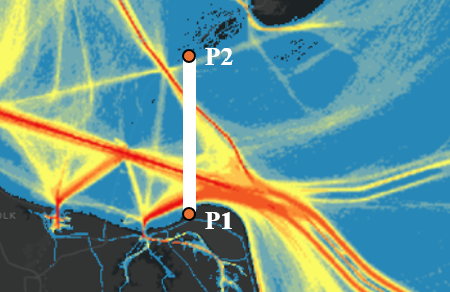}
\caption{Ship-traffic heatmap used for geometric validation of the inferred segment between $P_1$ and $P_2$.}
\label{fig:geometry}
\end{figure}

\textbf{Discretization and statistical model.}
The segment is discretized into \(G=80\) cells of width \(\Delta=L/G\approx\SI{0.223}{km}\). The observed crossing count in cell \(g\) is \(n_g\), and the latent log-intensity is \(f_g\). The model uses \(E_g = T_c \Delta\) with \(T_c=\SI{7}{days}\), and
\[
n_g \mid f_g \sim \mathrm{Poisson}(E_g e^{f_g})
\]
with a Mat\'ern-\(3/2\) covariance
\begin{align*}
K_{gh}
=
\sigma_f^2\left(1+\frac{\sqrt{3}|x_g-x_h|}{\ell}\right)
\exp\!\left(-\frac{\sqrt{3}|x_g-x_h|}{\ell}\right)
\end{align*}
The hyperparameters are set to \(\sigma_f=3.50\) and \(\ell=\SI{1.00}{km}\).

\begin{table}[t]
\caption{Main experimental configuration.}
\label{tab:config}
\centering
\small
\begin{tabular}{@{}ll@{}}
\toprule
Item & Setting \\
\midrule
Historical fitting window & \(T_c=\SI{7}{days}\) \\
Segment length & \(L=\SI{17.811}{km}\) \\
Segment discretization & \(G=80\) cells \\
Candidate spacing & \SI{0.05}{km} \\
Detection parameters & \(\rho=0.95,\ \sigma_\ell=0.90\) \\
Selected hyperparameters & \(\sigma_f=3.50,\ \ell=\SI{1.00}{km}\) \\
Planning horizon & \(T=\SI{1}{day}\) \\
Pathwise Monte Carlo bank & \(W_{\mathrm{main}}=1800\) \\
Ranking bank & \(W_{\mathrm{rank}}=350\) \\
Random prefixes & \(P=100\) \\
Paired endpoint bank & \(W_{\mathrm{paired}}=4000\) \\
\bottomrule
\end{tabular}
\end{table}

\textbf{Posterior characteristics.}
Fig.~\ref{fig:posterior} shows clear spatial heterogeneity in the posterior intensity estimated via INLA. This heterogeneity creates a regime in which uncertainty-aware approximations can differ meaningfully from the mean-only surrogate \(\nu_{\mathrm{J}}(\mathcal{A}) = \exp\!\left(-\mu_{Y_g}(\mathcal{A})\right)\) \cite{kim2023optimal}.

\begin{figure}[t]
\centering
\includegraphics[width=\linewidth]{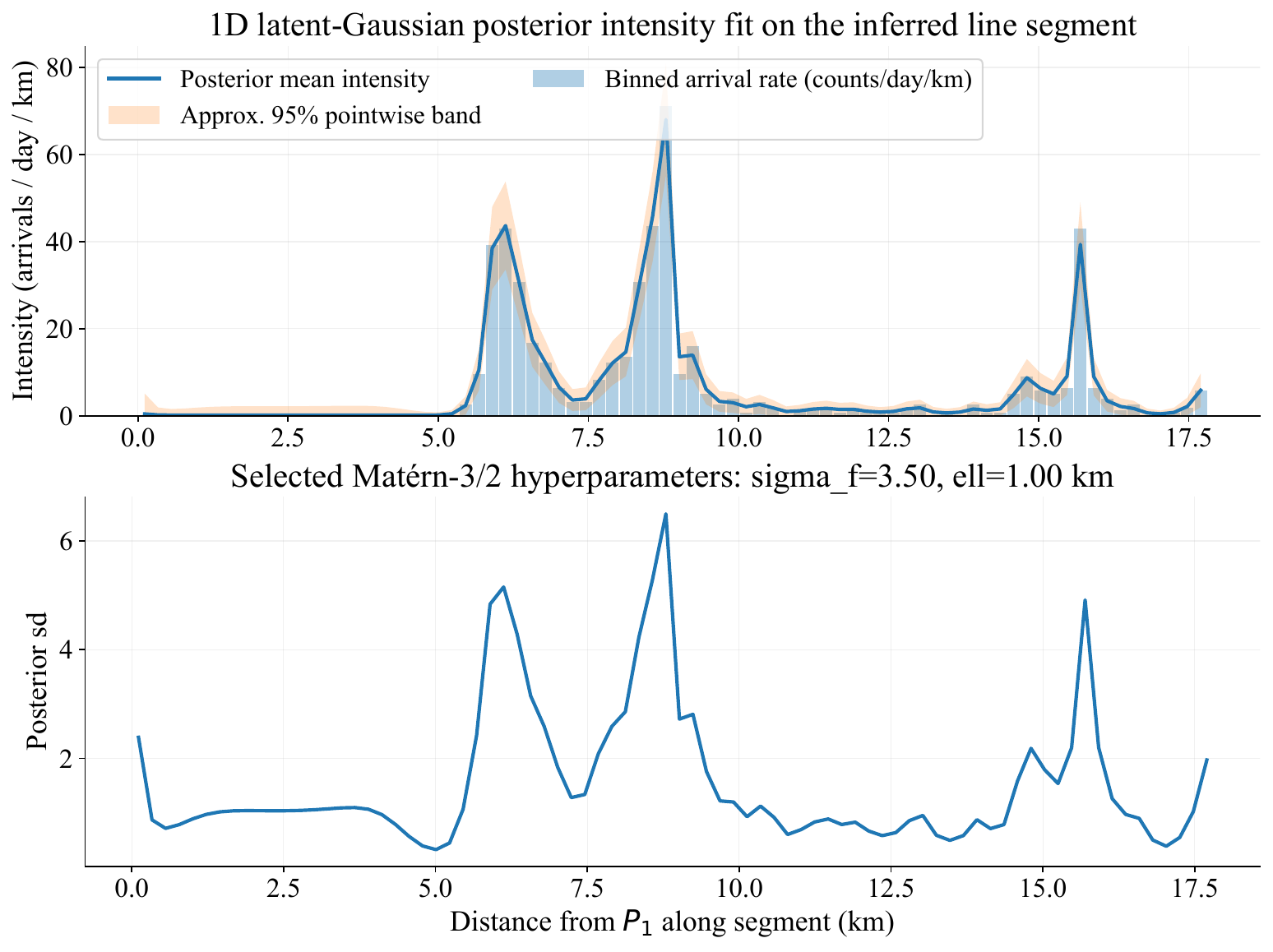}
\caption{Posterior fit on the inferred segment. Top: binned rate, posterior mean, and approximate pointwise band. Bottom: posterior standard deviation.}
\label{fig:posterior}
\end{figure}

\textbf{Evaluation setup and uncertainty scaling.}
The fitted posterior mean intensity is denoted by \(\bar\lambda\), and the latent covariance by \(\Sigma\). A mean-preserving lognormal scaling is used to control uncertainty while keeping the average intensity fixed. Gaussian draws \(Z^{(w)}\sim\mathcal{N}(0,\Sigma)\) define
\(m_c = \log\bar\lambda - \tfrac12 c\,\diag(\Sigma)\) and \(f^{(w,c)} = m_c + \sqrt{c}\,Z^{(w)}\), which lead to
\begin{align}
\lambda^{(w,c)} = \exp\bigl(f^{(w,c)}\bigr),
\qquad
\mathbb{E}[\lambda^{(w,c)}]=\bar\lambda
\label{eq:mean_preserving_b}
\end{align}
Three uncertainty regimes are examined: low uncertainty \((c=0.1)\), the fitted posterior \((c=1)\), and moderate uncertainty inflation \((c=10)\). Table~\ref{tab:config} summarizes the main experimental settings, where a 7-day window is used for fitting and a 1-day horizon is used for evaluation.

\textbf{Monte Carlo evaluation and metrics.}
Posterior Monte Carlo evaluation uses correlated samples generated from \eqref{eq:mean_preserving_b}. For a selected set \(\mathcal{A}\), the primary performance metrics are
\begin{align}
\widehat{\nu}(\mathcal{A}) &= \frac{1}{W}\sum_{w=1}^{W} e^{-Y_g^{(w)}(\mathcal{A})} \\
q_{0.05}(\mathcal{A}) &= \text{0.05-quantile of } \left\{ e^{-Y_g^{(w)}(\mathcal{A})} \right\}_{w=1}^{W}
\label{eq:q05}
\end{align}
where $X_g^{(w)}=T\Delta\lambda_g^{(w,c)}$ and $Y_g^{(w)}(\mathcal A)=\sum_g X_g^{(w)}\pi_g(\mathcal A)$ are correlated posterior realizations of the cell and grid aggregate exposures. The same grid definition is used to evaluate placements selected by all three surrogates. The main pathwise evaluation uses \(W_{\mathrm{main}}=1800\) samples to estimate the void probability \(\mathbb{E}[\exp(-Y_g(\mathcal{A}))]\) in \eqref{eq:discrete_void}. Ranking diagnostics use \(W_{\mathrm{rank}}=350\), common random-prefix diagnostics use \(P=100\) paths, and paired endpoint comparisons use a separate bank of \(W_{\mathrm{paired}}=4000\) samples.

\begin{table}[t]
\caption{Endpoint comparison at $K=20$. Means and quantiles use $W_{\rm main}$; paired differences and confidence intervals use $W_{\rm paired}$, so differences need not equal the displayed mean subtraction.}
\label{tab:main}
\centering
\scriptsize
\setlength{\tabcolsep}{3pt}
\begin{tabular}{@{}lcccccc@{}}
\toprule
Regime & $\widehat{\nu}_{\mathrm{J}}$ & $\widehat{\nu}_{\mathrm{GL}}$ & $\Delta_{\mathrm{GL}-\mathrm{J}}$ & 95\% CI & $q_{0.05,\mathrm{J}}$ & $q_{0.05,\mathrm{GL}}$ \\
\midrule
$c=0.1$  & 0.3097 & 0.3097 & 0.0000 & [0.0000, 0.0000] & 0.2558 & 0.2558 \\
$c=1$    & 0.3328 & 0.3400 & 0.0061 & [0.0042, 0.0079] & 0.1403 & 0.0712 \\
$c=10$   & 0.4500 & 0.5974 & 0.1478 & [0.1455, 0.1500] & 0.2044 & 0.3630 \\
\bottomrule
\end{tabular}
\end{table}

For the planning horizon \(T=\SI{1}{day}\), the corresponding surrogate moments are $\mu_g = T\Delta\,\bar\lambda_g$ and $
v_g = (T\Delta\,\bar\lambda_g)^2\bigl(e^{c\Sigma_{gg}}-1\bigr)$.

\begin{figure*}[t]
    \centering
    \includegraphics[width=0.88\textwidth]{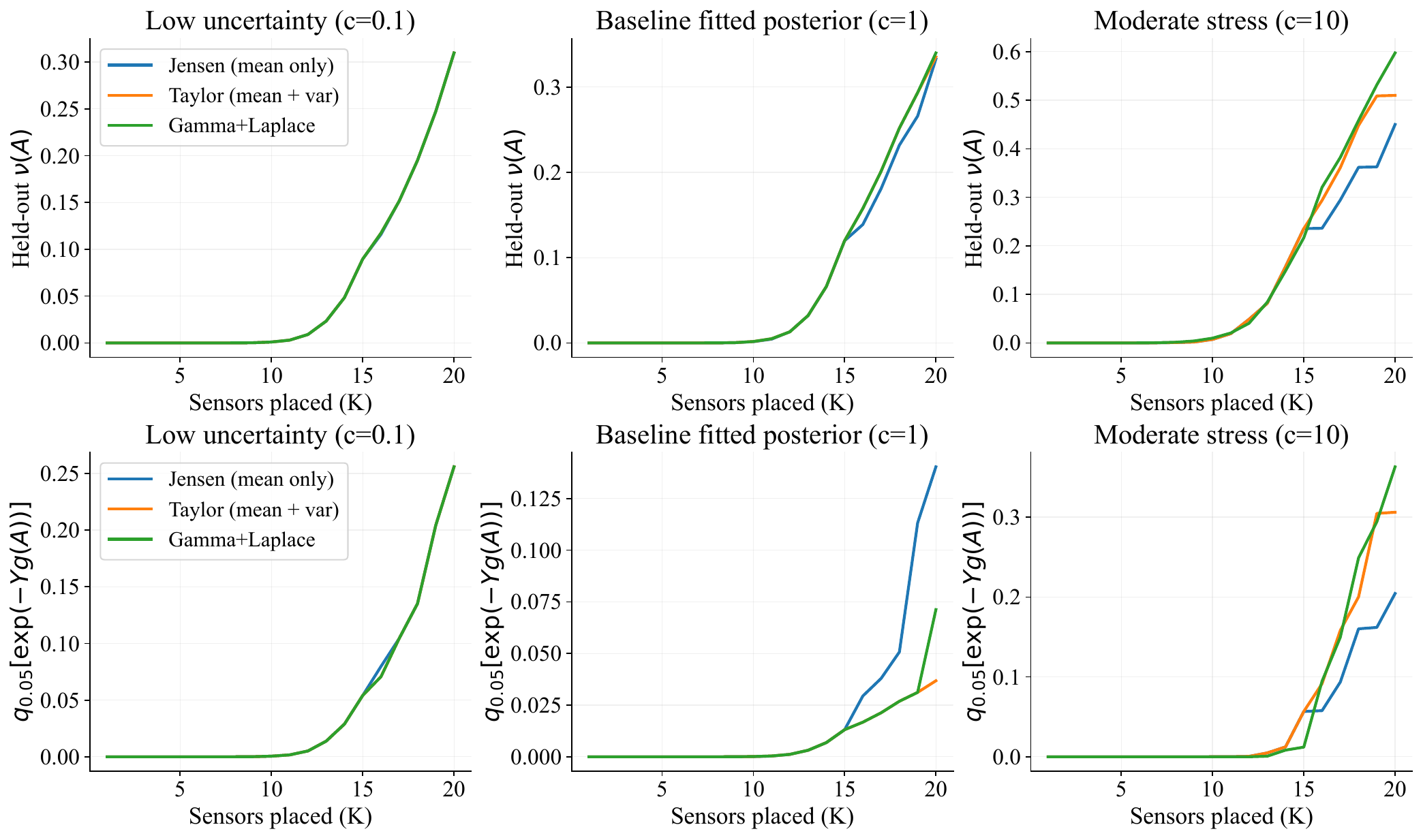}
    \caption{Posterior Monte Carlo greedy-path comparison. Top: $\widehat{\nu}(\mathcal A)$. Bottom: the 0.05-quantile of $e^{-Y_g(\mathcal A)}$. All methods are evaluated using correlated grid exposures.}
    \label{fig:greedy_paths}
\end{figure*}

\subsection{Greedy behavior and performance}
Table~\ref{tab:main} reports performance at \(K=20\). The low-uncertainty case serves as a check: Jensen and Gamma--Laplace coincide to numerical precision, consistent with Proposition~\ref{prop:zero_var}. At the fitted posterior (\(c=1\)), Gamma--Laplace modestly improves the mean value (from \(0.3328\) to \(0.3400\)), while the lower-tail metric decreases (from \(0.1403\) to \(0.0712\)). This indicates a tradeoff between average and lower-tail performance.

In the moderate-uncertainty regime (\(c=10\)), Gamma--Laplace improves both metrics simultaneously. The performance difference, evaluated using common random numbers (CRN), is \(0.1478\) with a 95\% confidence interval of \([0.1455, 0.1500]\), indicating that Gamma--Laplace consistently yields higher values of \(\widehat{\nu}\) than Jensen under identical Monte Carlo samples. The selected sensor sets also differ across regimes, sharing 20, 17, and 2 locations for \(c=0.1\), \(c=1\), and \(c=10\), respectively, so increasing uncertainty changes which locations are selected.

A second-order Taylor-based approximation \cite{kim2025variance} is also evaluated as a baseline. Table~\ref{tab:taylor_placeholder} reports its performance using the same Monte Carlo pipeline. It performs comparably in the fitted regime (\(c=1\)) but becomes unstable under increased uncertainty (\(c=10\)), with large approximation error and degraded ranking performance shown in Fig.~\ref{fig:rankdiag}.

\begin{table}[t]
\caption{Comparison of Taylor-based approximations using the ship-arrival Monte Carlo pipeline.}
\label{tab:taylor_placeholder}
\centering
\footnotesize
\setlength{\tabcolsep}{4pt}
\begin{tabular}{@{}lccc@{}}
\toprule
\textbf{Method (Regime)} & $\widehat{\nu}$ $(K=20)$ & \textbf{MAE} & \textbf{Mean Rank} \\
\midrule
Taylor ($c=1$)  & 0.3354 & $7.86\!\times\!10^{-5}$ & 6 \\
Taylor ($c=10$) & 0.5098 & $2.48\!\times\!10^{10}$ & 29.4 \\
\bottomrule
\end{tabular}
\end{table}

Fig.~\ref{fig:greedy_paths} shows that these improvements emerge throughout the greedy selection, not only at the final solution: the paths nearly coincide at \(c=0.1\), diverge slightly at \(c=1\), and the Gamma--Laplace trajectory clearly dominates both metrics at \(c=10\), particularly in later stages.

\begin{figure*}[t]
    \centering
    \includegraphics[width=0.76\textwidth]{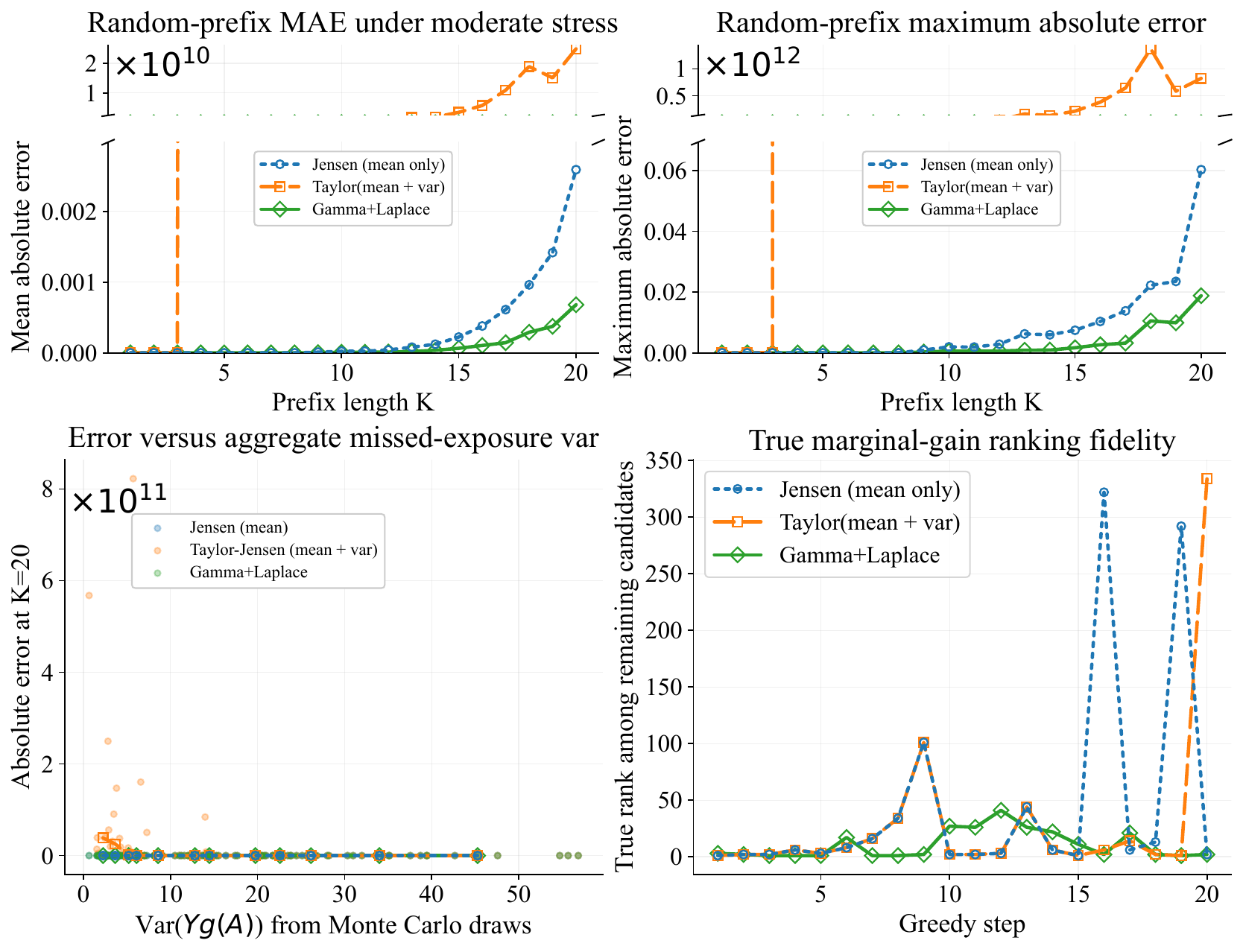}
    \caption{Moderate-regime prefix and ranking diagnostics at \(c=10\). Gamma--Laplace shows smaller common-prefix error and better late-stage ranking.}
    \label{fig:rankdiag}
\end{figure*}

\subsection{Surrogate fidelity and uncertainty effects}

Improved greedy performance is meaningful only if the surrogate accurately reflects the true objective. Table~\ref{tab:calibration} shows calibration and ranking diagnostics evaluated on common random placements or prefixes, leading to smaller absolute values than the metrics in Table~\ref{tab:main}. The results show that Gamma--Laplace is already better calibrated and provides more accurate candidate rankings at \(c=1\), and that this advantage increases substantially at \(c=10\).

\begin{table}[t]
\caption{Calibration and ranking fidelity at \(K=20\).}
\label{tab:calibration}
\centering
\scriptsize
\setlength{\tabcolsep}{3pt}
\begin{tabular}{@{}l l cccc@{}}
\toprule
Regime & Method & MAE & Slope & Mean rank & Spearman \\
\midrule
$c=1$  & Jensen & $2.56\!\times\!10^{-4}$ & 1.16  & 15.75 & 0.885 \\
$c=1$  & GL     & $4.90\!\times\!10^{-5}$ & 1.02  & 2.35  & 0.994 \\
$c=10$ & Jensen & $2.59\!\times\!10^{-3}$ & 2.32  & 43.30 & 0.735 \\
$c=10$ & GL     & $6.80\!\times\!10^{-4}$ & 0.845 & 10.50 & 0.972 \\
\bottomrule
\end{tabular}
\end{table}

\begin{figure*}
    \centering
    \includegraphics[width=0.76\textwidth]{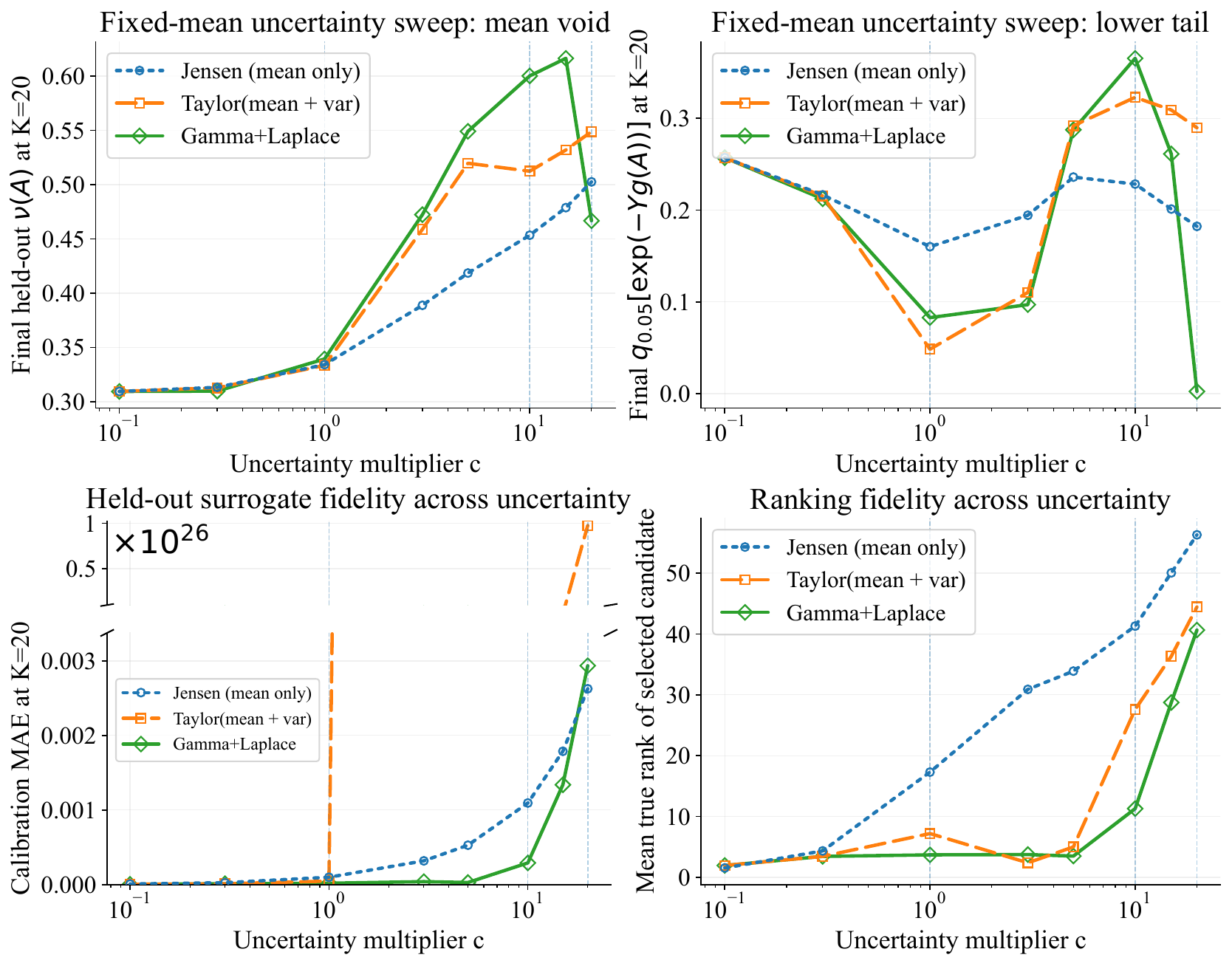}
    \caption{Fixed-mean uncertainty sweep. Top: final Monte Carlo mean void probability and lower-tail metric. Bottom: calibration MAE (broken axis) and mean selected-candidate rank. Points beyond the main $c\le10$ comparison show higher-uncertainty stress behavior, not uniform superiority.}
    \label{fig:sweep}
\end{figure*}

\section{Discussion}
\label{sec:discussion}

The results show that Gamma--Laplace improves surrogate fidelity and, consequently, greedy decisions. It yields lower mean absolute error, slopes closer to one (indicating better calibration), and stronger ranking metrics (mean rank and Spearman correlation). Because greedy selection depends on accurate marginal-gain ordering, these improvements explain the superior selection behavior observed in Fig.~\ref{fig:greedy_paths}.

In the fitted and moderate-uncertainty regimes, Gamma--Laplace matches the Monte Carlo objective more closely than Jensen. Proposition~\ref{prop:gamma_exact} establishes exactness only under the independent-Gamma model; agreement with the correlated Monte Carlo objective is an empirical result, not a consequence of that proposition. Fig.~\ref{fig:rankdiag} also shows that Gamma--Laplace accumulates less error along common prefixes and preserves the true ranking more accurately in later stages.

Fig.~\ref{fig:sweep} shows how the benefit of incorporating variance develops with increasing uncertainty over \(c \in \{0.1, 0.3, 1, 3, 5, 10\}\): Gamma--Laplace improves the mean metric beyond \(c=1\), while lower-tail improvements emerge at moderate stress. The additional higher-uncertainty points in the sweep show that the advantage is not uniform.

\begin{table}[t]
\caption{Robustness summary at $c=10$. Entries show the range of relative Gamma--Laplace gains across tested settings.}
\label{tab:robustness}
\centering
\scriptsize
\begin{tabular}{@{}lcc@{}}
\toprule
Group & Mean gain (\%) & $q_{0.05}$ gain (\%) \\
\midrule
Discretization & 11.8--36.3 & 11.3--111.1 \\
Sensor model & 16.9--41.7 & 8.4--68.9 \\
Planning horizon & 14.5--66.9 & 24.2--86.7 \\
Extraction rule & 35.0 & 93.2 \\
\bottomrule
\end{tabular}
\end{table}

\subsection{Robustness}

The moderate-uncertainty comparison remains positive across the tested sensitivity variations: discretizations \(G\in\{60,80,100\}\), sensor-kernel settings near \((\rho,\sigma_\ell)=(0.95,0.90)\), planning horizons \(T\in\{0.5,1,2\}\) days, and an extraction-rule variant that excludes points lying exactly on or collinear with the segment. Across all cases, Gamma--Laplace yields higher values of \(\widehat{\nu}\) and improved lower-tail performance. Table~\ref{tab:robustness} reports the corresponding ranges of relative gains at \(c=10\).

\subsection{Limitations and interpretation}

The comparison is not uniformly favorable: in the fitted-posterior regime, gains in the mean metric need not improve lower-tail behavior. As discussed under Assumption~\ref{ass:factorization}, the surrogate retains only cellwise means and marginal variances, so its guarantees apply to the surrogate objective rather than the exact correlated-posterior objective evaluated in Sec.~\ref{sec:numerical}.

The uncertainty multiplier \(c\) is a mean-preserving mechanism that isolates the effect of variance while keeping the fitted mean intensity fixed, rather than representing a validated posterior regime. In this sense, Jensen depends only on the mean, the second-order correction in \cite{kim2025variance} incorporates aggregate variance, and Gamma--Laplace captures how uncertainty is distributed across cells. In heterogeneous settings, this distinction can influence marginal-gain ranking and, consequently, sensor placement decisions.

\section{Conclusion}
Sensor placement to maximize void probability was studied using three surrogates: the Jensen mean-only approximation, a scalar Taylor mean--variance correction, and the proposed cell-structured Gamma--Laplace surrogate. After log transformation and a constant shift, Gamma--Laplace induces a monotone-submodular objective with the classical greedy approximation guarantee for the transformed problem. Jensen and Gamma--Laplace coincide in the low-uncertainty regime. In the fitted and moderate-uncertainty regimes, Gamma--Laplace tracks the Monte Carlo objective more closely and improves ranking accuracy and greedy selection behavior. The ship-traffic results show that incorporating cellwise uncertainty provides a practical advantage for sensor placement under uncertainty.

\bibliographystyle{IEEEtran}
\bibliography{main}

\end{document}